\documentclass[11pt]{article}	

\usepackage{amsmath,amssymb,amsfonts,bm}
\usepackage{amsthm}
\usepackage{algorithm}
\usepackage[noend]{algpseudocode}
\usepackage[colorlinks=true,allcolors=blue,pdfusetitle]{hyperref}
\usepackage{url, dsfont}
\usepackage{tabularx,booktabs,multirow}
\usepackage{xspace}
\usepackage{microtype}
\usepackage[nameinlink,capitalize]{cleveref}

\newtheorem{theorem}{Theorem}
\newtheorem{lemma}{Lemma}

\newtheorem{definition}[lemma]{Definition}
\newtheorem{problem}{Problem}

\makeatletter
\AddToHook{env/algorithmic/before}{\def\@currentcounter{ALG@line}}
\makeatother

\crefalias{ALG@line}{line}
\crefname{ALG@line}{line}{lines}
\Crefname{ALG@line}{Line}{Lines}

\newcommand{\citep}{\cite}
\newcommand{\citet}[2]{#2~\cite{#1}}

\usepackage[letterpaper,margin=1in]{geometry}

\newcommand{\IR}{\ensuremath{\mathbb{R}}}

\newcommand{\comm}[1]{}
\newcommand{\set}[1]{\left\{ #1 \right\}}

\newcommand{\provided}{\,|\,}

\newcommand{\OPT}{\ensuremath{\mathop{\mathrm{OPT}}}}

 \newcommand{\eps}{\varepsilon}

\newcommand{\etal}{{et~al.\/}}
\newcommand{\REM}[1]{}

\newcommand{\guess}{\mu}

\newcommand{\ve}{{\bm{e}}}
\newcommand{\1}{\bm{1}}

\renewcommand{\div}{\operatorname{div}}
\newcommand{\rep}{\operatorname{rep}}

\newcommand{\None}{\textsc{None}}
\newcommand{\FindRepair}{\textsc{FindRepair}\xspace}

\title{Fair Diversity Maximization via Local Search\thanks{
A preliminary version of this work was submitted to SODA~2026.
The current version was submitted on September 25, 2026 to ICLR~2027.
}}

\author{Mohammad Ansari \qquad Sina Azizeddin \qquad AmirMohammad Bandari \\ 
Pouria Mahmoudkhan \qquad Hamid Zarabi-Zadeh \\[1em]
Sharif University of Technology
}

\date{}

\begin{document}

\maketitle

\begin{abstract}
Diversity maximization is a fundamental optimization problem with applications in machine learning, data summarization, information retrieval, and recommendation systems. In many such applications, the data are partitioned into groups, and the selected subset must satisfy prescribed group quotas. We study Fair Diversity Maximization: given a set of points in a metric space partitioned into $m$ groups, the goal is to select exactly $k_i$ points from each group $i$ while maximizing the minimum pairwise distance among the selected points.

The best previously known approximation guarantee is $m+1$, which grows linearly with the number of groups. We show that this dependence on $m$ is not fundamental. We present a new local-search framework that yields a $4$-approximation for any constant number of groups, with no restrictions on the metric space or on the size of the selected set. To the best of our knowledge, this is the first constant-factor approximation whose guarantee is independent of the number of groups in this general setting. Our framework maintains all group quotas exactly while progressively eliminating violations of the diversity objective. We further develop a specialized algorithm for two groups that achieves a $2$-approximation, improving the previous best factor of $3$. This factor is optimal: unless $\mathrm{P}=\mathrm{NP}$, no polynomial-time algorithm can achieve an approximation factor strictly better than $2$, even for the unconstrained case.

\end{abstract}

% Keywords: Fairness, Diversity Maximization, Approximation Algorithms, Metric Spaces

% TL;DR: We develop new approximation algorithms for Fair Diversity Maximization that remove the linear dependence on the number of groups: a $4$-approximation for any constant number of groups and an optimal $2$-approximation for two groups.

% ----------------------------- Introduction --------------------------

\section{Introduction}

The increasing scale and complexity of modern datasets across domains such as machine learning have made it increasingly important to extract small yet informative data summaries. Such subset-selection problems arise whenever processing, storing, annotating, or presenting the entire dataset is computationally expensive. A basic challenge is then to decide which points are worth keeping. One natural principle is to favor diversity---choosing elements that are distinct from one another to obtain a varied and nonredundant view of the data. This simple idea has led to a broad line of work on diversity-based subset selection, with applications ranging from data summarization and information retrieval to recommendation systems, active learning, feature selection, training-data curation, and in-context demonstration selection~\citep{ahmed2019data,drosou2012disc,abbar2013real,kunaver2017diversity,zadeh2017scalable,sener2018active,liu2024tsds,wang2025demonstration}.

In the classical diversity maximization problem, also known as max-min diversity, the goal is to select a subset of a prescribed size that maximizes the minimum pairwise distance among the selected elements.
This simple objective has been studied for decades
with its roots in facility-location and dispersion problems~\citep{moon1984analysis,erkut1989analytical,erkut1990discrete,ravi1991facility}. 
Max-min diversification has also been used to construct diverse molecular libraries~\citep{ashton2002diverse}, and has recently appeared in machine-learning applications such as data subset selection and the pruning of redundant visual tokens in large multimodal and video-language models~\citep{fahrbach2025gist,alvar2025divprune,fan2026flashvid}.

In many applications, however, diversity alone is not enough. The objects in the underlying dataset often belong to different groups, and one may want the selected subset to contain a prescribed number of elements from each of them. These groups may correspond, for example, to demographic attributes, data sources, classes, genres, geographic regions, or time periods. An unconstrained diversity objective gives no guarantee that these groups will be adequately represented: a highly diverse solution may still contain very few or even entirely omit elements from a specific group. This issue arises naturally in applications such as training-data selection and pruning, where one may wish to preserve the representation of different classes or subpopulations, active learning, where only a limited number of points can be labeled, recommendation, where one may wish to control the representation of different genres or content providers, and data summarization, where the selected items may need to reflect different sources, categories, or temporal periods~\citep{rolf2021representation,hassan2025rcap,anahideh2022fair,celis2018fair,abbassi2013diversity,mahabadi2023core}.

In this paper, we study {Fair Diversity Maximization}. We are given a set of $n$ elements in a metric space, partitioned into $m$ disjoint groups, together with integers $k_1,\ldots,k_m$. The goal is to select exactly $k_i$ elements from group $i$, for each $i\in[m]$, so that the minimum distance between any two selected elements is as large as possible. Thus, the problem combines exact group-representation constraints with a max-min diversity objective.

The fair diversity maximization problem was first studied by \citet{moumoulidou2020diverse}{Moumoulidou~\etal}. They presented a $4$-approximation algorithm for the case of two groups, and extended their approach to obtain a $(3m-1)$-approximation for general $m$. 
\citet{addanki2022improved}{Addanki~\etal} subsequently
improved the approximation factor to $m+1$ when $m$ is constant. 
They also provided a 3-approximation for the case of two groups.

\subsection{Our Contributions}

All existing approximation guarantees for fair diversity maximization degrade with the number of groups~$m$. We show that this dependence on $m$ is not fundamental. In particular, we obtain a polynomial-time $4$-approximation for any constant number of groups, yielding the first constant-factor guarantee that is independent of $m$. Our result is stated formally below.

\begin{theorem}\label{thm:constant-groups}
	For any constant number of groups $m$, fair diversity maximization admits a polynomial-time $4$-approximation.
\end{theorem}

We further consider the important special case of two groups, where the
approximation guarantee can be improved to the optimal factor of $2$. This
improves upon the previous best factor of $3$ and matches the optimal
approximation factor for the unconstrained problem.

\begin{theorem}\label{thm:two-groups}
	Fair diversity maximization with two groups admits a polynomial-time $2$-approximation.
\end{theorem}

The approximation factor $2$ is the best possible, since achieving a factor better than
$2$ is NP-hard, even for the unconstrained problem.

A comparison of our results with the best previously known guarantees is provided in \Cref{table:summary}.

\begin{table}[t]
	\centering
	\caption{Summary of the results for fair diversity maximization.}
	\label{table:summary}
	\vspace*{1em}
\begin{tabularx}{.75\columnwidth}{>{\arraybackslash}Xll}
\toprule
\bf Number of & \bf Approximation \qquad\ & \multirow{2}{*}{\bf Reference }  \\
\bf Groups & \bf Factor  \\
\midrule
\multirow{3}{*}{$m = O(1)$} 
		& $3m - 1$     &  \citep{moumoulidou2020diverse}  \\ 
	    & $m+1$     & \citep{addanki2022improved} \\ 	
	    & $4$  & This work  \\ 
\midrule
\multirow{3}{*}{$m=2$ }    
		& $4$      & \citep{moumoulidou2020diverse}   \\ 
	    & $3$     & \citep{addanki2022improved} \\ 	
	    & $2$ (optimal)  & This work \\ 

\bottomrule
\end{tabularx}
\end{table}

\subsection{Technical Overview}

\paragraph{Going beyond the cluster-and-flow approach.}
Both the previous $3m-1$ and $m+1$ approximation algorithms follow the same general paradigm. They first organize points into proximity clusters,
with the intention of selecting at most one point from each cluster, and then use a flow computation to make this selection while satisfying the prescribed group counts exactly~\citep{moumoulidou2020diverse,addanki2022improved}. For this approach to
work, the clustering must satisfy two competing requirements: points chosen
from different clusters must be sufficiently far apart, so that selecting at most one point from each cluster guarantees diversity; at the same time,
there must remain enough clusters for such a selection to satisfy every group quota exactly.
The main difficulty is to satisfy these two requirements simultaneously. In both previous algorithms, doing so leads to clusters that may contain chains of nearby points spanning as many as \(m\) groups. As a result, the diameter of a cluster can grow linearly with \(m\), which is precisely where the \(m\)-dependence enters their approximation guarantees. 

Our approach, on the other hand, avoids trying to enforce diversity and fairness
simultaneously in one shot. Instead, we begin with a solution that already
satisfies all group quotas and gradually improve its diversity while
preserving fairness throughout. The difficulty is that even a local exchange
must preserve the exact group counts, and it is not clear whether an
appropriate local modification always exists. Establishing the existence of
such exchanges is the main technical challenge behind our framework.

\paragraph{A unified local-repair framework.}
Both our results are obtained through a unified local-search framework. We work with two distance scales $\alpha\leq\beta$, where $\alpha$ is the diversity we aim to achieve, while $\beta$ is a stronger separation maintained between points belonging to the same group. Starting from a fair solution satisfying this within-group separation, we repeatedly choose a selected point $q$ that has another selected point at distance less than $\alpha$. We then replace $q$, together with any points that conflict with the replacement, by a carefully chosen set of new points. We call such a set a \emph{$q$-repair}. Informally, a $q$-repair must be mutually compatible and must contain enough points from every group to compensate for all points from that group that are removed. Thus, applying a repair preserves both exact fairness and the within-group separation, while strictly decreasing the number of selected points that have at least one neighbor with distance less than $\alpha$. 
Consequently, the local search terminates with a fair solution of diversity at least $\alpha$.

\paragraph{Technical novelty.}
The main technical challenge is to prove that a suitable local improvement
always exists.
For any constant number of groups, this can be a delicate procedure, because inserting a single point may simultaneously force the removal
of points from several different groups. We capture these interactions using
\emph{repair vectors}, which record the surplus or deficit in each group caused
by an exchange. We show that a minimal collection of such vectors satisfying
the required balance has a simple structure and is exactly
balanced. A Steinitz-type vector-balancing argument then implies that this
collection has size bounded solely as a function of $m$. Consequently, for
constant $m$, a suitable local improvement can be found by exhaustive search
in polynomial time.

The two-group case allows for a sharper argument. Since there is only one
other group whose balance needs to be restored, we can greedily build a maximal
auxiliary set from that group. The maximality of this set gives us a useful
saturation property. We then compare this set with an optimal solution and use
the disjointness of suitable neighborhoods around optimal points to show, by a
counting argument, that one of these points can be used to form a valid local
improvement. This leads to the tight $2$-approximation.

\subsection{Related Work}

The diversity maximization problem was first studied in the 1980s in the context of facility location and dispersion~\citep{moon1984analysis,erkut1989analytical}.
\citet{erkut1990discrete}{Erkut} showed that the problem is NP-hard. 
Subsequently, \citet{ravi1991facility}{Ravi~\etal} 
gave a polynomial-time $2$-approximation in any metric
space and showed that no approximation factor better than $2$ is possible
unless $\mathrm{P}=\mathrm{NP}$. 

The fair variant of diversity maximization was introduced by \citet{moumoulidou2020diverse}{Moumoulidou~\etal}, who developed a $4$-approximation for two groups and a $(3m-1)$-approximation for general $m$. Subsequently, \citet{addanki2022improved}{Addanki~\etal} improved the guarantee for constant $m$ to $m+1$ and obtained a $3$-approximation for the two-group case.

Several stronger approximation guarantees have been obtained under additional assumptions, such as restricting the problem parameters, relaxing the fairness constraints, or imposing additional structure on the metric.
\citet{moumoulidou2020diverse}{Moumoulidou~\etal} obtained a
$5$-approximation when $m$ is constant and $k=o(\log n)$. 
\citet{addanki2022improved}{Addanki~\etal} showed that a $2$-approximation is possible
when the fairness constraints are required only in expectation, and gave a
$6$-approximation while selecting at least a $(1-\eps)$ fraction of the
required number of elements from every group. For Euclidean metrics, they also
showed that the problem is exactly solvable in one dimension and admits a
$(1+\eps)$-approximation in constant dimensions.

\citet{wang2023max}{Wang~\etal} considered a more general setting in which
each group has lower and upper representation bounds, obtaining a
$(5 + \eps)$-approximation in polynomial time when $m=O(1)$ and
$k=o(\log n)$. 
\citet{gaofair}{Gao and Borodin} improved the approximation
guarantees in refined and relaxed metric spaces, while \citet{kurkure2024faster}{Kurkure~\etal} developed constant-factor near-linear-time
algorithms for Euclidean instances. 
More recently, \citet{adriaens2025few}{Adriaens and Tatti} obtained improved guarantees when only a
constant number of representatives is required from each group, while \citet{zheng2026tabu}{Zheng~\etal} studied exact and metaheuristic methods for Fair
Max-Min Diversity.

Beyond Max-Min Diversity, several other distance-based diversity objectives
have been studied extensively. A prominent example is Max-Sum Diversity,
which maximizes the sum of pairwise distances among the selected elements.
\citet{ravi1991facility}{Ravi~\etal} gave a $4$-approximation for this objective, which
was later improved to a $2$-approximation by \citet{hassin1997approximation}{Hassin~\etal}. 
\citet{abbassi2013diversity}{Abbassi~\etal} proposed a local-search algorithm achieving
a tight $2$-approximation under partition-matroid constraints. 
\citet{borodin2017max}{Borodin~\etal} further generalized this setting by giving a
$2$-approximation for maximizing the sum of a monotone submodular function
and Max-Sum diversity. 
\citet{ceccarello2018fast,ceccarello2020general}{Ceccarello~\etal} subsequently developed
coreset-based algorithms for diversity maximization under matroid constraints.

A separate line of work has focused on diversity maximization in large-scale
and restricted computational models. 
\citet{indyk2014composable}{Indyk~\etal} introduced composable coresets with
constant-factor guarantees for several diversity objectives, and \citet{aghamolaei2015diversity}{Aghamolaei~\etal} subsequently improved the guarantees
for several of these objectives. Diversity maximization has also been studied
in MapReduce, streaming, and sliding-window settings~\citep{ceccarello2016mapreduce,borassi2019better}. For Max-Min Diversity in
particular, \citet{haqi2023almost}{Haqi and Zarrabi-Zadeh} obtained an
almost-tight $(2+\eps)$-approximation in the massively parallel computation
model, improving the previous factor of $6$ due to
\citet{indyk2014composable}{Indyk~\etal}. More recently, diversity maximization
has also been considered in fully dynamic settings~\citep{pellizzoni2025dynamic}.

Scalability has also been studied in the presence of group-representation
constraints. \citet{mahabadi2023core}{Mahabadi and Trajanovski} developed
coresets for fair Max-Sum and sum-of-nearest-neighbor diversity, 
while \citet{wang2023fair}{Wang~\etal} designed one-pass streaming and sliding-window
algorithms for Fair Max-Min Diversity. 
Most recently, \citet{ahmadvand2026composable}{Ahmadvand~\etal} developed composable coresets for fair diversity maximization.

Fairness and representation constraints have also been studied under other
selection objectives and notions of fairness. 
\citet{celis2018fair}{Celis~\etal} studied fair determinantal data summarization,
while \citet{wang2021fairrepresentative}{Wang~\etal} considered fair
representative subset selection under submodular objectives. More recently,
\citet{li2025individual}{Li and Wang} studied diversity maximization under
\emph{individual} rather than group fairness, obtaining bicriteria
approximation algorithms for Max-Min and Max-Sum diversity.
Related fairness constraints have also been considered in multiwinner voting
and submodular maximization~\citep{celis18,HalabiMNTT20}, as well as in
problems such as fair $k$-center clustering~\citep{chen2016matroid,kleindessner2019fair,jones2020fair}, fair set
cover~\citep{dehghankar2024fair}, and fair matching~\citep{garcia2020fair}.

% ---------------------------- Preliminaries --------------------------

\section{Preliminaries}
\label{sec:prel}

Let $(X, d)$ be a metric space, 
where $X$ is a set of $n$ points, 
and $d:X\times X\rightarrow \IR_{\geq 0}$ is a metric distance function. 
The notion of diversity is defined as follows.

\begin{definition}[Diversity]
Given a subset $S \subseteq X$, 
the \emph{diversity} of $S$, denoted by $\div(S)$, is 
the minimum pairwise distance between distinct elements in $S$, i.e.,
$$\div(S)=\min_{p,q\in S, p \neq q} d(p, q).$$
\end{definition}

The classical diversity maximization problem asks for a subset 
$S \subseteq X$ of $k$ points that maximizes $\div(S)$. 
The standard GMM algorithm~\citep{ravi1991facility}, 
presented in \Cref{alg:gmm}, achieves a $2$-approximation.

 \begin{algorithm}[ht]
     \caption{\textsc{GMM}$(X, k)$}
     \label{alg:gmm}
     \begin{algorithmic}[1]
         \Require A point set $X$ in a metric space and a parameter $k$
         \Ensure A subset $S\subseteq X$ of $k$ points
         \State $S \gets \{\text{an arbitrary point of } X\}$
         \For{$i = 2$ \textbf{to} $k$}
             \State find a point $p \in X \setminus S$ maximizing $d(p, S)$
             \State $S \gets S\cup \{p\}$
         \EndFor
         \State \Return $S$
     \end{algorithmic}
 \end{algorithm}

Incorporating fairness constraints, 
the main problem studied in this paper is formally defined as follows.

\begin{problem}[Fair Diversity Maximization]
	\label{problem: fdm}
	Given a set $X$ partitioned into $m$ disjoint groups $X_1, X_2, \ldots, X_m$ and positive integers $k_1, k_2, \ldots, k_m$, where $k = \sum_{i=1}^m k_i$, a subset $S \subseteq X$ is called \emph{fair} if $|S \cap X_i| = k_i$ for all $i \in [m]$.
	The goal is to find a fair subset $S \subseteq X$ of size $k$ that maximizes $\div(S)$.
\end{problem}

In the remainder of this section, we introduce the notation and tools used in our algorithms and analyses.
For each point $p \in X$, we denote the index of the group containing $p$ with $c(p)$. 
For a more intuitive representation, we refer to $c(p)$ as the \emph{color} of $p$.
Whenever a set is indexed, for example $S_i$, this means $S\cap X_i$,
i.e., the points in $S$ with color $i$.

\begin{definition}[Color Separation]
	A set $S\subseteq X$ is $\beta$-color separated if for every color $i\in [m]$
	$$\div(S_i) \ge \beta$$
\end{definition}

For set $S$ and point $q$ we define $S^{-q} = S \setminus \{q\}$.
We further define open balls and conflict sets as tools to enforce diversity.

\begin{definition}[Open Ball]
	Let $\beta \geq \alpha > 0$ be fixed parameters. 
	For each color $i\in [m]$ and point $p$ we define
	$$B_i(p)= \begin{cases} \{v\in X_i :d(v,p)<\alpha\}, & i\neq c(p),\\
						    \{v\in X_i :d(v,p)<\beta\}, & i = c(p), \end{cases}$$
	letting $B(p) = \bigcup_{i\in [m]} B_i(p)$ be the union over all colors.
\end{definition}

\begin{definition}[Conflict Set]
	Let $\beta \geq \alpha > 0$ be fixed parameters. 
	Given set $S$ and point $q\in S$, for each color $i\in [m]$ and point $p$ we define
	$$\Gamma_i^q(p,S) = B_i(p) \cap S^{-q},$$
	again letting $\Gamma^q(p,S) = \bigcup_{i\in [m]} \Gamma_i^q(p,S)$ be the union over all colors.
\end{definition}

When it is clear from context, we may remove the $S$ from $\Gamma_i^q(p,S)$.
We may further expand these definitions to include sets of points $P$ to denote the union,
i.e., $ \Gamma_i^q(P,S) =  \bigcup_{p\in P} \Gamma_i^q(p,S)$.
Moreover, %if $\alpha$ and $\beta$ are clear by context,
a point $p \in S$ is called \emph{bad} with respect to $S$ if $|S \cap B(p)| > 1$. 
In simpler terms, $p$ is called bad with respect to $S$ 
if there is a point in $S$ (other than $p$) that is ``too close'' to $p$.
It is important to note that if $S$ is $\beta$-color separated, 
a bad point may only be close to another point with a different color in $S$.

Throughout this paper, 
we denote by $S^*$ an optimal solution, 
and by $\OPT$ its diversity, i.e., $\div(S^*) = \OPT$. 
We refer to each point in $S^*$ as an \emph{optimal point}.

% ---------------------------- Framework ---------------------------

\section{Local-Repair Framework}
\label{sec:framework}

In this section, we present a unified local-repair framework 
that underlies our approximation algorithms. 
We use this framework in the next two sections to obtain 
a $4$-approximation for any constant number of groups, and
a tight $2$-approximation for the case of two groups.

Our local-repair framework is given in \Cref{alg: framework}. 
In this algorithm, we assume that we are given a value $\guess$ satisfying $\guess \leq \OPT$.
We are also given two lower bounds $\alpha \le \beta \le \frac{\guess}{2}$.

\begin{algorithm}[t]
    \caption{\sc Local-Repair Framework} \label{alg: framework}
    \begin{algorithmic}[1]
		
        \Require A point set $X$, a guess $\guess\le\OPT$, two values $\frac{\guess}{c} \le \alpha \le \beta \le \frac{\guess}{2}$, for some constant $c \geq 2$
        \Ensure A fair subset $S\subseteq X$ with $\div(S)\geq \alpha$.
		\For{$i\in [m]$}
        \State let $S_i = \text{\sc GMM}(X_i, k_i)$ 
        \EndFor
        \State let $S = \bigcup_{i\in [m]} S_i$ 
        \While {there exists bad point $q$ in $S$}
			\State $P\gets \textsc{FindRepair}(S,q,\alpha,\beta)$
			\State $ S \gets (S^{-q} \setminus  \Gamma^q(P)) \cup P $
		\EndWhile
		\State \Return $S$
    \end{algorithmic}
\end{algorithm}

We first construct a fair $\beta$-color separated set $S$ by 
applying the GMM algorithm (\Cref{alg:gmm})
independently to each color group 
to select $k_i$ points from color $i$.

We then start a local repair to gradually improve the diversity of $S$. 
In each iteration, we identify a bad point $q \in S$ and replace a subset of points in $S$, along with $q$
with another set of points such that the resulting solution remains fair while strictly decreasing the number of bad points. 
To facilitate this replacement, we introduce the notion of a $q$-repair set.

\begin{definition}[$q$-repair]
	Let $q\in S$ be a bad point. A set $P$ is a $q$-repair if.
	\begin{itemize}
		\item $\div(P) \ge \alpha$,
		\item $P$ is $\beta$-color separated,
		\item For every color $i\in [m]$, 
		$|\Gamma^q_i(P)| + \1_{[i=c(q)]} \le |P_i|$
	\end{itemize}
\end{definition}

The first two conditions ensure that no two points of $P$ violate the required
separation. 
Points of different colors are at distance at least $\alpha$, 
while
points of the same color are at distance at least $\beta$. 
The third condition
ensures that inserting $P$ preserves the group quotas. 
After removing $q$ and
all points in $\Gamma^q(P)$ from $S$ and inserting $P$, 
the resulting set
contains at least $k_i$ points of every color $i$. 
This also explains the
terminology \emph{conflict set}, since the points in $\Gamma^q(P)$ are exactly
the selected points 
that prevent $P$ from being inserted while maintaining the
required separation.

The \textsc{FindRepair} subroutine is responsible for finding a suitable $q$-repair set.
This is the setting-specific part of the framework, and its implementation can be adapted to the problem at hand.
Specifically, 
the user of this framework must design a proper \FindRepair subroutine
and appropriate parameters $\alpha$ and $\beta$ depending on $\guess$
such that \FindRepair always
returns a $q$-repair whenever $\guess \leq \OPT$.

\begin{theorem}
Given a polynomial-time subroutine $\textsc{FindRepair}$, a guess $\guess\leq\OPT$, and appropriately chosen parameters $\alpha$ and $\beta$ for $\textsc{FindRepair}$ satisfying $\frac{\guess}{c}\leq\alpha\leq\beta\leq\frac{\guess}{2}$ for some constant $c\geq 2$, \Cref{alg: framework} finds a fair set $S$ with $\div(S)\geq\frac{\guess}{c}$ in polynomial time.
\end{theorem}

\begin{proof}
The correctness of the framework follows from a simple invariant. 
Initially, GMM guarantees diversity at least $\guess/2\geq\beta$ within each color
group, so $S$ is $\beta$-color separated. In each iteration, let $P$ be the
$q$-repair returned by \textsc{FindRepair}. By the definition of
$\Gamma^q(P)$, removing $q$ and $\Gamma^q(P)$ eliminates every point that
would conflict with a point of $P$.

 Since $P$ has diversity at least $\alpha$
and is itself $\beta$-color separated, inserting $P$ creates no new bad
points, while $q$ is removed.

The balance condition in the definition of a $q$-repair guarantees that, for
every color $i$,
\[
|S_i|-|\Gamma_i^q(P)|-\1_{[i=c(q)]}+|P_i|\geq k_i.
\]
Hence, any excess points can be removed while preserving both fairness and
$\beta$-color-separation. Therefore, the number of bad points strictly
decreases in every iteration. Since there are at most $k$ bad points, the
algorithm terminates after at most $k$ iterations, at which point
$\div(S)\geq\alpha \geq \guess / c$.

It remains to explain how the value $\mu$ is chosen. Let
$
D:=\{d(u,v) \provided u,v\in X\}
$
be the set of all pairwise distances.
Since $\OPT\in D$, we may run \Cref{alg: framework} for every
$\mu\in D$ and return the solution with maximum diversity.
Alternatively, we can use binary search over $D$ to find the desired solution more efficiently.
Since for every
guess $\mu\leq\OPT$ the algorithm returns a fair solution of diversity at least
$\alpha \geq \mu / c$, the run with $\mu=\OPT$ yields a $c$-approximation.
\end{proof}

% ---------------------------- Constant Groups ---------------------------

\section{4-Approximation for any Constant $m$}

In this section we show how to use the local-repair framework developed in the previous section to obtain a 4-approximation for fair diversity maximization for any constant number of groups.
To this end, we need to design a \textsc{FindRepair} subroutine 
that finds a $q$-repair for $(\alpha,\beta)=(\frac\guess4,\frac\guess2)$
in polynomial time.

We first introduce the repair vector of $p$.

\begin{definition}
	For a point $p$, we define \emph{repair vector} $\rep(p) \in \mathbb{Z}^m$ to be
	$$\rep_i(p)= |B_i(p)\cap S^{-q}| - \1_{[i=c(q)]}.$$
\end{definition}
The interpretation is that inserting $p$ into $S^{-q}$ supplies one unit of its own color while consuming one
unit for every conflicting point in $S$ that must be removed.

\begin{lemma}\label{lemma: blocker disjoint}
	For any two optimal points $p_1,p_2 \in S^*$,
	the sets $B(p_1)$ and $B(p_2)$ are disjoint.
\end{lemma}
\begin{proof}
	Suppose there exists a $v \in B(p_1) \cap B(p_2)$. Then,
	$$\guess \le d(p_1,p_2) \le d(p_1,v) + d(v,p_2) < \frac12\guess + \frac12\guess,$$
	which is a contradiction.
\end{proof}

\begin{lemma}
	For the optimal set $S^*$ we have
	$$\sum_{p\in S^*} \rep(p) \le - \ve^{(c(q))},$$
	where $e_j$ is the $j$-th unit vector, and the inequality is coordinate-wise. 
\end{lemma}
\begin{proof}
	According to \Cref{lemma: blocker disjoint} each $v\in S^{-q}$
	is inside at most one $B(p)$ for $p\in S^*$.
	This means that for every color $i$,
	$$\sum_{p\in S^*} |B_i(p)\cap S^{-q}| \le |S^{-q}_i| = k_i - \1_{[i=c(q)]}.$$
	Since $|S^*_i| = k_i$ then we must have
	$$\sum_{p\in S^*} \rep_i(p) = \sum_{p\in S^*} |B_i(p)\cap S^{-q}| - |S^*_i| \le - \1_{[i=c(q)]},$$
	which is exactly $-\ve^{(c(q))}$.
\end{proof}

This lemma shows that $S^*$ is a $q$-repair for all $S$ and $q$.
Let $P^* \subseteq S^*$ be the smallest subset of $S^*$ satisfying
$$\sum_{p\in P^*} \rep(p) \le - \ve^{(c(q))}.$$
We will refer to this inequality as the \emph{balance inequality}.
In the following lemmas, we will show that $|P^*|$ is bounded by a function of only $m$,
independent of $k$ and $n$.

\begin{lemma}\label{lemma: vector binary}
			Given $p\in P^*$, for every color $i$, we must have $|B_i(p)\cap S^{-q}| \le 1$,
	and in particular, $|B_{c(p)}(p)\cap S^{-q}| = 0$.
\end{lemma}

\begin{proof}
First suppose $i\neq c(p)$ and
$|B_i(p)\cap S^{-q}|\ge 2$.
Pick distinct $u,v\in B_i(p)\cap S^{-q}$.
Since $u,v$ have the same color,
$$
\frac12\guess\le d(u,v)
\le d(u,p)+d(v,p)
<\frac14\guess+\frac14\guess,
$$
a contradiction.

Now let $i=c(p)$ and suppose
$|B_i(p)\cap S^{-q}|\ge 1$.
Then all coordinates of $\rep(p)$ are non-negative, and hence
$$
\sum_{p'\in P^*\setminus\{p\}}\rep(p')
\le -\rep(p)-\ve^{(c(q))}
\le -\ve^{(c(q))}.
$$
Thus $P^*\setminus\{p\}$ still satisfies the balance inequality,
contradicting the minimality of $P^*$.
\end{proof}

An important corollary is that for all $p\in P^*$,
$B_i(p)\cap S^{-q}$ is either empty or contains a single point,
and therefore, $\rep(p) \in \set{-1,0,1}^m$.
We now describe a useful property of $P^*$:
the inequality in the balance property is indeed an equality.

\begin{algorithm}[t]
    \caption{\sc FindRepair General}
    \label{alg:3repair}
    \begin{algorithmic}[1]
		
        \Require A $\frac{\guess}{2}$-color separated set $S$, a point $q \in S$,
        and two parameters $(\alpha,\beta) = (\frac{\guess}{4},\frac{\guess}{2})$
        \Ensure A $q$-repair set $P$, provided $\guess \leq \OPT$
		
		\For {all $P\subseteq X$ with size at most $c_m$} 			
			\If {$P$ is a valid $q$-repair set}
				\Return $P$
			\EndIf
		\EndFor
		
		\State \Return \None

    \end{algorithmic}
\end{algorithm}

\begin{lemma}\label{lemma: et equal}
	The minimal set $P^*$ satisfying the balance inequality must satisfy
	$$
	\sum_{p\in P^*}\rep(p)=-\ve^{(c(q))}.
	$$
\end{lemma}
\vspace*{-1.8em}
\begin{proof}
	Suppose the inequality is strict in some coordinate $i$. Since the vectors
	are integral,
	$$
	\sum_{p\in P^*}\rep_i(p)\le -1-\1_{[i=c(q)]}.
	$$
	By \Cref{lemma: vector binary}, the only negative entry of $\rep(p)$ is
	$\rep_{c(p)}(p)=-1$. Hence, there exists $p_0\in P^*$ with $c(p_0)=i$.

	For every $j\neq i$, $\rep_j(p_0)\ge0$, so removing $p_0$ preserves the
	balance inequality in coordinate $j$. For coordinate $i$,
	$$
	\sum_{p\in P^*\setminus\{p_0\}}\rep_i(p)
	=
	\sum_{p\in P^*}\rep_i(p)+1
	\le -\1_{[i=c(q)]}.
	$$
	Thus $P^*\setminus\{p_0\}$ still satisfies the balance inequality,
	contradicting the minimality of $P^*$.
\end{proof}

Finally we have enough structure to prove our claim

\begin{lemma}\label{lemma: bound on Pstar}
	There exists a constant $c_m$ (depending only on $m$), such that 
	the minimal subset $P^* \subseteq S^*$ satisfying the balance inequality 
	has at most $c_m$ points.
\end{lemma}

\begin{proof}
	Consider the vectors $\{\rep(p):p\in P^*\}$ together with $\ve^{(c(q))}$.
	By \Cref{lemma: et equal}, their sum is zero, and by
	\Cref{lemma: vector binary}, each lies in $\{-1,0,1\}^m$ and hence has
	$\ell_\infty$-norm at most $1$.

	By the Steinitz lemma, they can be ordered as $v_1,\ldots,v_t$ so that every
	partial sum $u_r=\sum_{j=1}^r v_j$ lies in $[-m,m]^m$.
	Thus, there are at most $(2m+1)^m$ possible values for the partial sums.

	If $|P^*|>(2m+1)^m-1$, then $t=|P^*|+1>(2m+1)^m$, so two partial sums
	coincide. The vectors between them form a nonempty zero-sum set, and so do
	the remaining vectors. The set not containing $\ve^{(c(q))}$ therefore consists
	of vectors $\rep(p)$ for a nonempty subset $P'\subseteq P^*$ whose sum is
	zero. Hence,
	$$
	\sum_{p\in P^*\setminus P'}\rep(p)=-\ve^{(c(q))},
	$$
	contradicting the minimality of $P^*$. 
	Setting $c_m = (2m+1)^m - 1$ completes the proof.
\end{proof}

By \Cref{lemma: bound on Pstar}, whenever $\guess\le\OPT$, there exists a $q$-repair set of size at most $c_m$. Since $m$ is constant, $c_m$ is also constant, and \Cref{alg:3repair} can enumerate all subsets of size at most $c_m$ and check in polynomial time whether any is a valid $q$-repair. Thus, \Cref{alg:3repair} runs in polynomial time. Plugging this algorithm into our local-repair framework yields a $4$-approximation for fair diversity maximization, proving \Cref{thm:constant-groups}.

% ---------------------------- Two Groups ---------------------------

\section{Tight Bound for Two Groups}
\label{sec:two-groups}

In this section, we present a simple algorithm for finding a $q$-repair with $(\alpha,\beta)=(\frac{\guess}{2},\frac{\guess}{2})$. Combined with our local-repair framework, this algorithm yields a $2$-approximation for fair diversity maximization with two groups.

Throughout this section, we use red and blue to denote the two groups for simplicity.
Accordingly, we denote the corresponding subsets of $X$
by $X_r$ and $X_b$, and the associated color quotas
by $k_r$ and $k_b$, respectively.

\begin{algorithm}[t]
    \caption{\sc FindRepair Two-Groups}
    \label{alg:2repair}
    \begin{algorithmic}[1]
		
        \Require A $\frac{\guess}{2}$-color separated set $S$, a point $q \in S$,
        and two parameters $(\alpha,\beta) = (\frac{\guess}{2},\frac{\guess}{2})$
        \Ensure A $q$-repair $P$, provided $\guess \leq \OPT$

		\State rename colors so that $c(q)$ is red
		
		\If {there exists a red point $p$ outside $B(S^{-q})$} \label{line:2rep-5}
			\State\Return $\set{p}$\label{line:2rep-ret1}
		\EndIf
		
		\State let $R = \set{}$
		\While {a point $v$ exists outiside $B(R\cup S^{-q})$}
			\State $R \gets R \cup \set{v}$ \label{line:2rep-7}
		\EndWhile
		\State choose a red point $p\not\in S^{-q}$ with $|B(p) \cap S^{-q}_r| = 0$
									and $|B(p) \cap (S_b^{-q}\cup R)| \le |R|$ \label{line:2rep-10}
				
        \State\Return $\set{p} \cup (R\setminus B(p))$\label{line:2rep-ret2}
    \end{algorithmic}
\end{algorithm}

Our implementation of $\FindRepair$ for the two-color case is given in \Cref{alg:2repair}. The algorithm takes as input $(\alpha,\beta)=(\frac{\guess}{2},\frac{\guess}{2})$, a $\frac{\guess}{2}$-color-separated set $S$, and a point $q\in S$.
The following three lemmas show that the algorithm correctly returns a $q$-repair,
provided $\guess\le\OPT$.

\begin{lemma}
\label{lemma: same ball}
If $\guess \le \OPT$ and $\alpha \le \beta = \frac{\guess}{2}$, then for any point $q$, we have $|B(q)\cap S^*|\le 1$.
\end{lemma}

\begin{proof}
Suppose, for contradiction, that there exist distinct points $p_1,p_2\in B(q)\cap S^*$. Then
\[
       \OPT \le d(p_1,p_2)
        \le d(p_1,q)+d(q,p_2)
        < 2\beta
        = \guess
        \le \OPT,
\]
which is a contradiction.
\end{proof}

A direct consequence of this lemma is that the sets $B(v)$ are pairwise disjoint for all $v\in S^*$. Intuitively, the points in $S^*$ are pairwise at distance at least $\OPT$, while each ball $B(v)$ has radius at most $\frac{1}{2}\OPT$.

\begin{lemma}
\label{lemma: p exists}
If $\guess \le \OPT$, then a point $p$ satisfying the conditions in \cref{line:2rep-10} of \Cref{alg:2repair} exists.
\end{lemma}

\begin{proof}
Define $x=|S_b^{-q} \cup R|-k_b$. Since the condition in \cref{line:2rep-5} is not met, all points added to $R$ are blue. Hence $x=|R|$. Moreover, since $q$ is red, we have $|S_r^{-q}|=k_r-1$.

We first show that there exists a red point $p\in S_r^*$ such that
$|B(p)\cap S_r^{-q}|=0$. By \cref{lemma: same ball}, for every $q_1\in S_r^{-q}$,
\[
    |B(q_1)\cap S_r^*|
    \le |B(q_1)\cap S^*|
    \le 1.
\]
Therefore,
\[
    |B(S_r^{-q})\cap S_r^*|
    \le \sum_{q_1\in S_r^{-q}} |B(q_1)\cap S_r^*|
    \le k_r-1.
\]
Since $|S_r^*|=k_r$, there exists some $p\in S_r^*$ such that
$|B(p)\cap S_r^{-q}|=0$.

It remains to show that this point $p$ satisfies
$|B(p)\cap(S_b^{-q}\cup R)|\le x$. By \cref{lemma: same ball}, the sets
$B(v)\cap(S^{-q}\cup R)$ are pairwise disjoint for $v\in S^*$. Moreover, by the construction of $R$, every point $v$ of $S^*$ has at least one point in $B(v)\cap(S^{-q}\cup R)$; otherwise, it would have been added to $R$. Thus,
\begin{align*}
    k_b+k_r+x-1
    &= |S^{-q}\cup R| \\
    &\ge |B(S^*)\cap(S^{-q}\cup R)| \\
    &= \sum_{v\in S^*}|B(v)\cap(S^{-q}\cup R)| \\
    &\ge |S^*\setminus\{p\}|+|B(p)\cap(S^{-q}\cup R)| \\
    &= k_b+k_r-1+|B(p)\cap(S^{-q}\cup R)|.
\end{align*}
Hence
\[
    |B(p)\cap(S^{-q}\cup R)|\le x.
\]
Since $|B(p)\cap S_r^{-q}|=0$, this implies
\[
    |B(p)\cap(S_b^{-q}\cup R)|\le x=|R|.
\]
Thus, $p$ satisfies both required properties. Finally, such a point can be found in polynomial time by checking every red point.
\end{proof}

\begin{lemma}
\label{lemma: 2-repair}
If $\guess\le \OPT$, then \Cref{alg:2repair} returns a valid $q$-repair in polynomial time.
\end{lemma}

\begin{proof}
The algorithm terminates in either \cref{line:2rep-ret1} or \cref{line:2rep-ret2}.
First, suppose that the algorithm terminates in \cref{line:2rep-ret1}. Then $P=\set{p}$, where $p\notin B(S^{-q})$. Hence, $|\Gamma_i^q(P)|=0$ for both colors $i$. Since $p$ and $q$ have the same color, we have $|P_r|=1$, and thus the color condition for red is satisfied. The color condition for blue is immediate, since $P_b=\emptyset$ and $\Gamma_b^q(P)=\emptyset$. Moreover, since $P$ contains only one point, it is trivially $(\frac{\guess}{2},\frac{\guess}{2})$-compatible. Thus, $P$ is a valid $q$-repair.

Now suppose that the algorithm terminates in \cref{line:2rep-ret2}. By \Cref{lemma: p exists}, there exists a desirable point $p$ satisfying
\[
    |B(p)\cap S_r^{-q}|=0
    \qquad\text{and}\qquad
    |B(p)\cap(S_b^{-q}\cup R)|\le |R|.
\]
The algorithm returns
\[
    P=\set{p}\cup(R\setminus B(p)).
\]
By the construction of $R$, $B(R)$ contains no point of $S^{-q}$. Moreover, $B(p)\cap S^{-q}$ contains only blue points. Therefore,
\[
    |\Gamma_r^q(P)|=0.
\]
Since $|P_r|=1$, the red color condition follows:
\[
    |P_r|=1\ge 1+|\Gamma_r^q(P)|.
\]

For the blue color, we have
\begin{align*}
    |P_b|
    &=|R\setminus B(p)|\\
    &=|R|-|B(p)\cap R|\\
    &=|R|+|B(p)\cap S_b^{-q}|
      -|B(p)\cap(S_b^{-q}\cup R)|\\
    &\ge |B(p)\cap S_b^{-q}|.
\end{align*}
Since $p\in P$ and $R\setminus B(p)$ contains no point of $B(p)$,
\[
    |B(p)\cap S_b^{-q}|
    =|B(P)\cap S_b^{-q}|
    =|\Gamma_b^q(P)|.
\]
Hence
\[
    |P_b|\ge |\Gamma_b^q(P)|,
\]
and the blue color condition is satisfied.

It remains to verify compatibility. By construction, the points in $R$ are mutually compatible and have no point of $S^{-q}$ in their respective balls. When $p$ is added to $P$, all points in $B(p)$ are removed from $R$. Thus, every point remaining in $R\setminus B(p)$ is compatible with $p$, and consequently $P$ is $(\frac{\guess}{2},\frac{\guess}{2})$-compatible.

Therefore, in either case, \Cref{alg:2repair} returns a valid $q$-repair set. All steps of the algorithm, including the required distance and neighborhood checks, can be performed in polynomial time.
\end{proof}

Plugging \Cref{alg:2repair} into our local-repair framework yields a $2$-approximation for fair diversity maximization with two groups, thereby proving \Cref{thm:two-groups}.

% ---------------------------- Conclusion --------------------------

\section{Conclusion}

We studied fair diversity maximization under group-representation
constraints and showed that the previously known dependence on the number of
groups can be avoided when $m$ is an arbitrary constant. 
Our main result is a polynomial-time
$4$-approximation for any constant number of groups, with an approximation
guarantee independent of $m$. We also obtained a tight $2$-approximation for
the special case of two groups.
Both results are based on a unified local-repair framework that preserves the
group quotas exactly while gradually improving the diversity objective. 

Having eliminated the dependence on $m$ from the approximation factor, a natural direction for future work is to determine whether the factor $4$ can be further improved. In particular, our tight $2$-approximation for two groups raises the question of whether the optimal factor of $2$ can be achieved for any constant number of groups. Establishing such a tight approximation guarantee even for small values of $m$ remains an intriguing open problem.

% ---------------------------- AI Use Statement --------------------------

\subsection*{AI use statement}

We used AI-assisted tools to improve the clarity, grammar, and presentation of the manuscript. In particular, these tools were used for proofreading and rephrasing portions of the text. All mathematical results, algorithms, proofs, and claims were developed and verified by the authors. The authors take full responsibility for the content of the paper.

% ---------------------------- Bibliography --------------------------

\newcommand{\etalchar}[1]{$^{#1}$}


\begin{thebibliography}{ZMGMG26}

\bibitem[AARZZ26]{ahmadvand2026composable}
Ali Ahmadvand, Mohammad Ansari, Mobin Razavi, and Hamid Zarrabi-Zadeh.
\newblock Composable coresets for fair diversity maximization.
\newblock In {\em Proceedings of the 38th ACM Symposium on Parallelism in
  Algorithms and Architectures}, pages 561--570, 2026.

\bibitem[AAT22]{anahideh2022fair}
Hadis Anahideh, Abolfazl Asudeh, and Saravanan Thirumuruganathan.
\newblock Fair active learning.
\newblock {\em Expert Systems with Applications}, 199:116981, 2022.

\bibitem[AAYIM13]{abbar2013real}
Sofiane Abbar, Sihem Amer-Yahia, Piotr Indyk, and Sepideh Mahabadi.
\newblock Real-time recommendation of diverse related articles.
\newblock In {\em Proceedings of the 22nd International Conference on World
  Wide Web (WWW)}, pages 1--12, 2013.

\bibitem[ABC{\etalchar{+}}02]{ashton2002diverse}
Mark Ashton, John Barnard, Florence Casset, Michael Charlton, Geoffrey Downs,
  Dominique Gorse, John Holliday, Roger Lahana, and Peter Willett.
\newblock Identification of diverse database subsets using property-based and
  fragment-based molecular descriptions.
\newblock {\em Quantitative Structure-Activity Relationships}, 21(6):598--604,
  2002.

\bibitem[AFZZ15]{aghamolaei2015diversity}
Sepideh Aghamolaei, Majid Farhadi, and Hamid Zarrabi-Zadeh.
\newblock Diversity maximization via composable coresets.
\newblock In {\em Proceedings of the 27th Canadian Conference on Computational
  Geometry}, pages 38--48, 2015.

\bibitem[Ahm19]{ahmed2019data}
Mohiuddin Ahmed.
\newblock Data summarization: a survey.
\newblock {\em Knowledge and Information Systems}, 58(2):249--273, 2019.

\bibitem[AMMM22]{addanki2022improved}
Raghavendra Addanki, Andrew McGregor, Alexandra Meliou, and Zafeiria
  Moumoulidou.
\newblock Improved approximation and scalability for fair max-min
  diversification.
\newblock In {\em Proceedings of the 25th International Conference on Database
  Theory}, pages 7:1--7:21, 2022.

\bibitem[AMT13]{abbassi2013diversity}
Zeinab Abbassi, Vahab~S Mirrokni, and Mayur Thakur.
\newblock Diversity maximization under matroid constraints.
\newblock In {\em Proceedings of the 19th ACM SIGKDD International Conference
  on Knowledge Discovery and Data Mining}, pages 32--40, 2013.

\bibitem[ASAZ25]{alvar2025divprune}
Saeed~Ranjbar Alvar, Gursimran Singh, Mohammad Akbari, and Yong Zhang.
\newblock {DivPrune}: Diversity-based visual token pruning for large multimodal
  models.
\newblock In {\em Proceedings of the IEEE/CVF Conference on Computer Vision and
  Pattern Recognition}, pages 9392--9401, 2025.

\bibitem[AT25]{adriaens2025few}
Florian Adriaens and Nikolaj Tatti.
\newblock Fair diversity maximization with few representatives.
\newblock In {\em Proceedings of the 31st ACM SIGKDD International Conference
  on Knowledge Discovery and Data Mining}, pages 17--25, 2025.

\bibitem[BEL{\etalchar{+}}19]{borassi2019better}
Michele Borassi, Alessandro Epasto, Silvio Lattanzi, Sergei Vassilvitskii, and
  Morteza Zadimoghaddam.
\newblock Better sliding window algorithms to maximize subadditive and
  diversity objectives.
\newblock In {\em Proceedings of the 38th ACM Symposium on Principles of
  Database Systems (PODS)}, pages 254--268, 2019.

\bibitem[BJLY17]{borodin2017max}
Allan Borodin, Aadhar Jain, Hyun~Chul Lee, and Yuli Ye.
\newblock Max-sum diversification, monotone submodular functions, and dynamic
  updates.
\newblock {\em ACM Transactions on Algorithms}, 13(3):41:1--41:25, 2017.

\bibitem[CHV18]{celis18}
L.~Elisa Celis, Lingxiao Huang, and Nisheeth~K. Vishnoi.
\newblock Multiwinner voting with fairness constraints.
\newblock In {\em Proceedings of the 27th International Joint Conference on
  Artificial Intelligence}, pages 144--151, 2018.

\bibitem[CKS{\etalchar{+}}18]{celis2018fair}
Elisa Celis, Vijay Keswani, Damian Straszak, Amit Deshpande, Tarun Kathuria,
  and Nisheeth Vishnoi.
\newblock Fair and diverse {DPP}-based data summarization.
\newblock In {\em Proceedings of the 37th International Conference on Machine
  Learning (ICML)}, pages 716--725, 2018.

\bibitem[CLLW16]{chen2016matroid}
Danny~Z Chen, Jian Li, Hongyu Liang, and Haitao Wang.
\newblock Matroid and knapsack center problems.
\newblock {\em Algorithmica}, 75:27--52, 2016.

\bibitem[CPP18]{ceccarello2018fast}
Matteo Ceccarello, Andrea Pietracaprina, and Geppino Pucci.
\newblock Fast coreset-based diversity maximization under matroid constraints.
\newblock In {\em Proceedings of the 11th ACM International Conference on Web
  Search and Data Mining}, pages 81--89, 2018.

\bibitem[CPP20]{ceccarello2020general}
Matteo Ceccarello, Andrea Pietracaprina, and Geppino Pucci.
\newblock A general coreset-based approach to diversity maximization under
  matroid constraints.
\newblock {\em ACM Transactions on Knowledge Discovery from Data}, 14(5):1--27,
  2020.

\bibitem[CPPU17]{ceccarello2016mapreduce}
Matteo Ceccarello, Andrea Pietracaprina, Geppino Pucci, and Eli Upfal.
\newblock Mapreduce and streaming algorithms for diversity maximization in
  metric spaces of bounded doubling dimension.
\newblock {\em Proceedings of the VLDB Endowment}, 10(5):469--480, 2017.

\bibitem[DP10]{drosou2012disc}
Marina Drosou and Evaggelia Pitoura.
\newblock Search result diversification.
\newblock {\em ACM SIGMOD Record}, 39(1):41--47, 2010.

\bibitem[DRSA24]{dehghankar2024fair}
Mohsen Dehghankar, Rahul Raychaudhury, Stavros Sintos, and Abolfazl Asudeh.
\newblock Fair set cover.
\newblock In {\em Proceedings of the 31st ACM SIGKDD International Conference
  on Knowledge Discovery and Data Mining}, pages 189--200, 2024.

\bibitem[EN89]{erkut1989analytical}
Erhan Erkut and Susan Neuman.
\newblock Analytical models for locating undesirable facilities.
\newblock {\em European Journal of Operational Research}, 40(3):275--291, 1989.

\bibitem[Erk90]{erkut1990discrete}
Erhan Erkut.
\newblock The discrete $p$-dispersion problem.
\newblock {\em European Journal of Operational Research}, 46(1):48--60, 1990.

\bibitem[FCX{\etalchar{+}}26]{fan2026flashvid}
Ziyang Fan, Keyu Chen, Ruilong Xing, Yulin Li, Li~Jiang, and Zhuotao Tian.
\newblock {FlashVID}: Efficient video large language models via training-free
  tree-based spatiotemporal token merging.
\newblock In {\em International Conference on Learning Representations (ICLR)
  (Poster)}, 2026.

\bibitem[FRZ{\etalchar{+}}25]{fahrbach2025gist}
Matthew Fahrbach, Srikumar Ramalingam, Morteza Zadimoghaddam, Sara Ahmadian,
  Gui Citovsky, and Giulia DeSalvo.
\newblock {GIST}: Greedy independent set thresholding for max-min
  diversification with submodular utility.
\newblock In {\em Advances in Neural Information Processing Systems (NeurIPS)},
  volume~38, pages 46322--46349, 2025.

\bibitem[GB24]{gaofair}
Ashley~Qianxi Gao and Allan Borodin.
\newblock Fair max-min diversification in refined and relaxed metric spaces.
\newblock Preprint, 2024.

\bibitem[GSB20]{garcia2020fair}
David Garc{\'\i}a-Soriano and Francesco Bonchi.
\newblock Fair-by-design matching.
\newblock {\em Data Mining and Knowledge Discovery}, 34:1291--1335, 2020.

\bibitem[HKP25]{hassan2025rcap}
Atif Hassan, Swanand Khare, and Jiaul~H. Paik.
\newblock {RCAP}: Robust, class-aware, probabilistic dynamic dataset pruning.
\newblock In {\em Proceedings of the Forty-first Conference on Uncertainty in
  Artificial Intelligence}, volume 286 of {\em Proceedings of Machine Learning
  Research}, pages 1648--1662. PMLR, 2025.

\bibitem[HMN{\etalchar{+}}20]{HalabiMNTT20}
Marwa~El Halabi, Slobodan Mitrovic, Ashkan Norouzi{-}Fard, Jakab Tardos, and
  Jakub Tarnawski.
\newblock Fairness in streaming submodular maximization: Algorithms and
  hardness.
\newblock In {\em Proceedings of the 33rd Advances in Neural Information
  Processing Systems (NeurIPS)}, pages 9150--9171, 2020.

\bibitem[HRT97]{hassin1997approximation}
Refael Hassin, Shlomi Rubinstein, and Arie Tamir.
\newblock Approximation algorithms for maximum dispersion.
\newblock {\em Operations research letters}, 21(3):133--137, 1997.

\bibitem[HZZ23]{haqi2023almost}
Alireza Haqi and Hamid Zarrabi-Zadeh.
\newblock Almost optimal massively parallel algorithms for $k$-center
  clustering and diversity maximization.
\newblock In {\em Proceedings of the 35th ACM Symposium on Parallel Algorithms
  and Architectures}, pages 239--247, 2023.

\bibitem[IMMM14]{indyk2014composable}
Piotr Indyk, Sepideh Mahabadi, Mohammad Mahdian, and Vahab~S Mirrokni.
\newblock Composable core-sets for diversity and coverage maximization.
\newblock In {\em Proceedings of the 33rd ACM Symposium on Principles of
  Database Systems (PODS)}, pages 100--108, 2014.

\bibitem[JNN20]{jones2020fair}
Matthew Jones, Huy Nguyen, and Thy Nguyen.
\newblock Fair $k$-centers via maximum matching.
\newblock In {\em Proceedings of the 39th International Conference on Machine
  Learning (ICML)}, pages 4940--4949, 2020.

\bibitem[KAM19]{kleindessner2019fair}
Matth{\"a}us Kleindessner, Pranjal Awasthi, and Jamie Morgenstern.
\newblock Fair $k$-center clustering for data summarization.
\newblock In {\em Proceedings of the 38th International Conference on Machine
  Learning (ICML)}, pages 3448--3457, 2019.

\bibitem[KP17]{kunaver2017diversity}
Matev{\v{z}} Kunaver and Toma{\v{z}} Po{\v{z}}rl.
\newblock Diversity in recommender systems: {A survey}.
\newblock {\em Knowledge-Based Systems}, 123:154--162, 2017.

\bibitem[KSW{\etalchar{+}}24]{kurkure2024faster}
Yash Kurkure, Miles Shamo, Joseph Wiseman, Sainyam Galhotra, and Stavros
  Sintos.
\newblock Faster algorithms for fair max-min diversification in
  {$\mathrm{R}^{\mathrm{d}}$}.
\newblock {\em Proceedings of the ACM on Management of Data}, 2(3):1--26, 2024.

\bibitem[LKR24]{liu2024tsds}
Zifan Liu, Amin Karbasi, and Theodoros Rekatsinas.
\newblock {TSDS}: Data selection for task-specific model finetuning.
\newblock In {\em Advances in Neural Information Processing Systems (NeurIPS)},
  volume~37, pages 10117--10147, 2024.

\bibitem[LW25]{li2025individual}
Ruien Li and Yanhao Wang.
\newblock Individually fair diversity maximization.
\newblock In {\em Advances in Neural Information Processing Systems (NeurIPS)},
  volume~38, pages 160253--160278, 2025.

\bibitem[MC84]{moon1984analysis}
I~Douglas Moon and Sohail~S Chaudhry.
\newblock An analysis of network location problems with distance constraints.
\newblock {\em Management Science}, 30(3):290--307, 1984.

\bibitem[MMM21]{moumoulidou2020diverse}
Zafeiria Moumoulidou, Andrew McGregor, and Alexandra Meliou.
\newblock Diverse data selection under fairness constraints.
\newblock In {\em Proceedings of the 24th International Conference on Database
  Theory}, pages 13:1--13:25, 2021.

\bibitem[MT23]{mahabadi2023core}
Sepideh Mahabadi and Stojan Trajanovski.
\newblock Core-sets for fair and diverse data summarization.
\newblock In {\em Advances in Neural Information Processing Systems (NeurIPS)},
  volume~36, pages 78987--79011, 2023.

\bibitem[PPP25]{pellizzoni2025dynamic}
Paolo Pellizzoni, Andrea Pietracaprina, and Geppino Pucci.
\newblock Fully dynamic clustering and diversity maximization in doubling
  metrics.
\newblock {\em ACM Transactions on Knowledge Discovery from Data}, 19(4):1--45,
  2025.

\bibitem[RRT91]{ravi1991facility}
Sekharipuram~S Ravi, Daniel~J Rosenkrantz, and Giri~Kumar Tayi.
\newblock Facility dispersion problems: Heuristics and special cases.
\newblock In {\em Proceedings of the 17th Workshop on Algorithms and Data
  Structures}, pages 355--366, 1991.

\bibitem[RWRJ21]{rolf2021representation}
Esther Rolf, Theodora~T. Worledge, Benjamin Recht, and Michael~I. Jordan.
\newblock Representation matters: Assessing the importance of subgroup
  allocations in training data.
\newblock In {\em Proceedings of the 38th International Conference on Machine
  Learning (ICML)}, volume 139 of {\em Proceedings of Machine Learning
  Research}, pages 9040--9051. PMLR, 2021.

\bibitem[SS18]{sener2018active}
Ozan Sener and Silvio Savarese.
\newblock Active learning for convolutional neural networks: A core-set
  approach.
\newblock In {\em International Conference on Learning Representations (ICLR)
  (Poster)}, 2018.

\bibitem[WFM21]{wang2021fairrepresentative}
Yanhao Wang, Francesco Fabbri, and Michael Mathioudakis.
\newblock Fair and representative subset selection from data streams.
\newblock In {\em Proceedings of the Web Conference 2021}, pages 1340--1350,
  2021.

\bibitem[WFML23]{wang2023fair}
Yanhao Wang, Francesco Fabbri, Michael Mathioudakis, and Jia Li.
\newblock Fair max-min diversity maximization in streaming and sliding-window
  models.
\newblock {\em Entropy}, 25(7):1066, 2023.

\bibitem[WMLF23]{wang2023max}
Yanhao Wang, Michael Mathioudakis, Jia Li, and Francesco Fabbri.
\newblock Max-min diversification with fairness constraints: Exact and
  approximation algorithms.
\newblock In {\em Proceedings of the 2023 SIAM International Conference on Data
  Mining}, pages 91--99, 2023.

\bibitem[WWY{\etalchar{+}}25]{wang2025demonstration}
Xubin Wang, Jianfei Wu, Yuan Yichen, Deyu Cai, Mingzhe Li, and Weijia Jia.
\newblock Demonstration selection for in-context learning via reinforcement
  learning.
\newblock In {\em Proceedings of the 42nd International Conference on Machine
  Learning (ICML)}, volume 267 of {\em Proceedings of Machine Learning
  Research}, pages 64937--64954. PMLR, 2025.

\bibitem[ZGMZ17]{zadeh2017scalable}
Sepehr Zadeh, Mehrdad Ghadiri, Vahab Mirrokni, and Morteza Zadimoghaddam.
\newblock Scalable feature selection via distributed diversity maximization.
\newblock In {\em Proceedings of the 31st AAAI Conference on Artificial
  Intelligence}, pages 2876--2883, 2017.

\bibitem[ZMGMG26]{zheng2026tabu}
He~Zheng, Anna Mart{\'i}nez-Gavara, Rafael Mart{\'i}, and Fred Glover.
\newblock Advanced probabilistic tabu search for the fair max-min diversity
  problem.
\newblock {\em Computers \& Operations Research}, 196:107650, 2026.

\end{thebibliography}
\end{document}